\documentclass[article, a4]{amsart}
\usepackage {amsmath, amssymb, color, textcomp, amsfonts}
\usepackage{mathtools}

\usepackage[top=1in, bottom=1in, left=1in, right=1in]{geometry}

\makeatletter
\newtheorem{thm}{Theorem}[section]

\newtheorem{lem}[thm]{Lemma}

\newtheorem{prop}[thm]{Proposition}

\newtheorem{preremark}[thm]{Remark}
\newenvironment{remark}%
  {\begin{preremark}\upshape}{\end{preremark}}
\newtheorem{preexample}[thm]{Example}
\newenvironment{example}%
  {\begin{preexample}\upshape}{\end{preexample}}

\numberwithin{equation}{section}

\numberwithin{equation}{section}

\numberwithin{thm}{section}
\newtheorem{lem/defn}[thm]{Lemma/Definition}
\newtheorem{preex/defn}[thm]{Example/Definition}
\newenvironment{ex/defn}%
  {\begin{preex/defn}\upshape}{\end{preex/defn}}

\newcommand{\ten}{\otimes}

\newcommand{\abs}[1]{\lvert#1\rvert}

\begin{document}

\title[The second bosonization of the CKP hierarchy]{The second bosonization of the CKP hierarchy}
\author{Iana I. Anguelova}

\address{Department of Mathematics,  College of Charleston,
Charleston SC 29424 }
\email{anguelovai@cofc.edu}

\date{\today}


\begin{abstract}
In this paper we discuss the second bosonization of the Hirota bilinear equation for the CKP hierarchy  introduced in \cite{DJKM6}. We show that there is  a second, untwisted,  Heisenberg action  on the Fock space,  in addition to the  twisted Heisenberg action suggested by \cite{DJKM6} and studied in \cite{OrlovLeur}. We derive the decomposition of  the Fock space into irreducible   Heisenberg modules under this action. We show that the vector space spanned by the  highest weight vectors of the irreducible Heisenberg modules has a structure of a super vertex algebra, specifically the symplectic fermions vertex algebra. We complete the second bosonization of the CKP Hirota equation by expressing the generating field via exponentiated boson vertex operators acting on a polynomial algebra with two infinite sets of variables.
\end{abstract}

\maketitle


\section{Introduction}
\label{sec:intro}

The Kadomtsev-Petviashvili (KP) hierarchy is famously associated with the boson-fermion correspondence, a vertex algebra isomorphism between the charged free fermions super vertex algebra and the  lattice super vertex algebra of the rank one odd lattice (see e.g. \cite{Kac}). One of the aspects of the boson-fermion correspondence is the equivalence between the KP hierarchy of differential equations in the bosonic space and the algebraic Hirota bilinear equation on the fermionic space. Namely, the KP hierarchy can be defined by the following Hirota bilinear equation:
\[
Res_z \Big(\psi^+ (z)\otimes \psi^- (z)\Big) (\tau \ten \tau) =0,
\]
where $\psi^+ (z)$ and  $\psi ^-(z)$ are the two fermionic fields generating the charged free fermions super vertex algebra (following the notation of \cite{Kac}), and $\tau$ is an element of the Fock space of states of this super vertex algebra (the charge 0 subspace, to be exact). But the KP hierarchy is a hierarchy of differential equations, hence to demonstrate the equivalence with the Hirota bilinear approach one needs to bosonize the fields $\psi^+ (z)$ and $\psi ^-(z)$, i.e., write them in terms of bosonic (differential) operators. This  bosonization was one side of the isomorphism known as the boson-fermion correspondence (there is a vast literature on this, as well as other aspects of the boson-fermion correspondence, see e.g. \cite{KacRaina}, \cite{Kac}, \cite{Miwa-book} among many others).

In \cite{DJKM-4} and \cite{DJKM6}    Date, Jimbo, Kashiwara and Miwa  introduced two  new hierarchies related to the KP hierarchy: the BKP and the CKP hierarchies.  As was the case  for the  KP hierarchy as well, the BKP and the CKP hierarchies were initially defined via a Lax form instead of Hirota bilinear equation:
  \[
  \frac{\partial L}{\partial x_n} =[(L^n)_+, L],
  \]
   where $L$ is a certain pseudo-differential operator of the form $L=\partial +u_1(x)\partial^{-1} + u_1(x)\partial^{-2}+\dots $ (see e.g. \cite {DJKM-KP}, \cite{Miwa-book} for details). The connection between the Hirota bilinear equation and the Lax forms is given by
   \[
   u_1 = \frac{\partial ^2}{\partial x_1} \ln \tau (x).
   \]
   Specifically, both the BKP and the CKP hierarchies were defined as  reductions from the KP hierarchy, by assuming conditions on the pseudo-differential operator $L$ used in the Lax form.
  For both of them Date, Jimbo, Kashiwara and Miwa suggested a Hirota bilinear equation, i.e., operator approach.
   The Hirota equation approach  was later completed for the BKP hierarchy (see  \cite{YouBKP}, among others).  There  were no surprises encountered for  the BKP case, and similarly to the KP case   the bosonization of the BKP hierarchy was shown to be  one of the sides of the boson-fermion correspondence of type B (\cite{DJKM-4}, \cite{YouBKP}), which was later interpreted as an isomorphism of certain twisted vertex (chiral)  algebras (\cite{Ang-Varna2}, \cite{AngTVA}).

In \cite{DJKM6}    Date, Jimbo, Kashiwara and Miwa suggested the following Hirota equation for the CKP hierarchy:
\[
Res_z \big(\chi (z)\otimes \chi (-z)\big) (\tau \ten \tau) =0,
\]
where the field $\chi (z)$ is  actually itself bosonic,
with OPE
\[
\chi(z)\chi(w)\sim \frac{1}{z+w}.
\]
Even though the field $\chi (z)$ is bosonic (and thus the algebra generated by its operator coefficients is a Lie algebra, see Section 1), we still need to bosonize it further in terms of  Heisenberg algebra operators, in order to recover the connection with the Lax approach.
In \cite{DJKM6} Date, Jimbo, Kashiwara and Miwa suggested an approach to bosonization, via a twisted Heisenberg field defined by the normal ordered product $\frac{1}{2}:\chi (z) \chi (-z) :$, but did not complete the bosonization.
In \cite{OrlovLeur} van de Leur, Orlov and Shiota  completed this suggested  bosonization and derived further properties and applications.
The CKP hierarchy though held several surprises, with more yet to come perhaps. The most consequential one so far, and the one we address in this paper, is that  the CKP hierarchy admits two different actions of  two different Heisenberg algebras,  one twisted and one untwisted, and thus two bosonizations of the Hirota equation are possible.  The existence of these two different Heisenberg actions  was discovered in \cite{AngMLB}. The twisted Heisenberg algebra was the one used in \cite{OrlovLeur}. In this paper we complete the second bosonization, initiated by the second, untwisted, Heisenberg algebra action. We will study the further properties and applications of this bosonization in a consequent paper, but here in this paper we concentrate on the necessary steps to complete the bosonization.

There are 3 stages to any bosonization:
\begin{enumerate}
\item Construct a bosonic Heisenberg current from the generating fields, hence obtaining a field representation of the Heisenberg algebra on the Fock space;
\item Decompose the Fock space into irreducible Heisenberg modules;
\item Express the original generating fields in terms of exponential boson fields, if possible.
\end{enumerate}
 This paper  completes  these three stages. In Section 1 we introduce the required notation, recall the two Heisenberg actions (further on in this paper we will  only be concerned with the untwisted Heisenberg action), and introduce two necessary gradings. We then follow through with the decomposition of the Fock space into irreducible Heisenberg modules. Herein lies  the second surprise of the CKP hierarchy: although similarly to the KP case there is a charge decomposition of the Fock space (via the charge grading induced by the action of $h_0$, the 0 component of the untwisted Heisenberg field), unlike for the KP Fock space the charge decomposition is not the same as the decomposition into irreducible modules. Specifically, unlike in the KP case, none of the charge components is irreducible as a Heisenberg module. Here indeed the Fock space is  completely reducible, but the vector space spanned by the highest weight vectors of the Heisenberg modules has a much more detailed and fine structure. (This is true for the first bosonization completed in \cite{OrlovLeur} as well). We show in Proposition \ref{prop:HeisDecomp} that the indexing set $\mathfrak{P}_{tdo}$ for the highest weight vectors (and thus for the irreducible Heisenberg modules in the decomposition) consists of the  distinct partitions with a triangular part plus a distinct subpartition of odd half integers, namely:
\[
\mathfrak{P}_{tdo}=\{\mathfrak{p} =(T_m, \lambda_1, \lambda_2, \dots , \lambda_k) \ |\ T_m-\text{triangular\ number}, \ \lambda_1>\lambda_2>\dots >\lambda_k, \lambda_i \in \frac{1}{2}+\mathbb{Z}_{\geq 0}, \ i=1, \dots , k \}.
\]
Further, as we show  in Section 3, the space of  highest weight vectors  has a structure of a super vertex algebra, and  specifically a structure realizing the symplectic fermion vertex super algebra, introduced first in \cite{Curiosities} and \cite{SymplFermionsKausch} (see also triplet vertex algebra, and e.g. \cite{Wangbosonization}, \cite{Abe}, \cite{ADAMOVIC}). Finally, in Section 4, by using the known embedding of the symplectic fermion vertex algebra as a subalgebra (the "small" subalgebra) of  the charged free fermion vertex algebra (following \cite{FMS}, \cite{Curiosities}), and thus via the boson-fermion correspondence into a lattice vertex algebra,  we can express the  generating field $\chi (z)$ via exponentiated boson vertex operators acting on a polynomial algebra with two infinite sets of variables.

\section{Heisenberg action and module decomposition}
\label{Heis-decomposition}

We will use common concepts and technical tools from the areas of vertex algebras and conformal field theory, such as the notions of field, locality, Operator Product Expansions (OPEs), normal ordered products, Wick's Theorem, for which we refer the reader to any book on the topic (such as \cite{FLM} and \cite{Kac}). We will also use  the extension of these technical tools to the case of $N$-point locality, as introduced in \cite{ACJ}.

The starting point is the twisted neutral boson field $\chi (z)$,
\begin{equation}
\chi (z) = \sum _{n\in \mathbb{Z}+1/2} \chi _n z^{-n-1/2}
\end{equation}
with OPE
\begin{equation}
\label{equation:OPE-C}
\chi(z)\chi(w)\sim \frac{1}{z+w}.
\end{equation}
This OPE determines the commutation relations between the modes $\chi_n$, $n\in \mathbb{Z} +1/2$, as \footnote{ We use here an indexing of  the field $\chi (z)$ typical of vertex algebra fields (as opposed to   \cite{DJKM6}, where it was introduced  originally,  or  \cite{OrlovLeur}).  The field  $\chi (z)$ is  related to the double-infinite rank Lie algebra $c_{\infty}$  (see e.g. \cite{WangKac}, \cite{WangDuality}, \cite{ACJ}); consequently it is denoted by $\phi ^C (z)$ in \cite{ACJ}.}:
\begin{equation}
\label{eqn:Com-C}
[\chi_m, \chi_n]=(-1)^{m-\frac{1}{2}}\delta _{m, -n}1.
\end{equation}

The modes of the field $\chi (z)$ form a Lie algebra which we denote by $L_{\chi}$.  Let $\mathit{F_{\chi}}$ be  the Fock module of $L_{\chi}$ with  vacuum vector $|0\rangle $, such that $\chi_n|0\rangle=0 \ \text{for} \ \  n > 0$.
The vector space $\mathit{F_{\chi}}$ has a basis
\begin{equation}
\{\left(\chi _{-j_k}\right)^{m_k}\dots \left(\chi _{-j_2}\right)^{m_2}\left(\chi _{-j_1}\right)^{m_1}|0\rangle \  \arrowvert \ \ j_k>\dots > j_2 > j_1 > 0, \ j_i\in \mathbb{Z}+\frac{1}{2}, \ m_i > 0, m_i\in \mathbb{Z}, \ i=1, 2, \dots, k\}.
\end{equation}
Thus  with our indexing $\mathit{F_{\chi}}$ is isomorphic to  the Fock space $\mathit{F}$ of \cite{OrlovLeur}).

 In \cite{DJKM6} Date, Jimbo, Kashiwara and Miwa introduced the CKP hierarchy through a reduction of the KP hierarchy,  and suggested the following algebraic Hirota bilinear equation associated to it:
\begin{equation} \label{eqn:Hirotaeqn}
Res_z \big(\chi (z)\otimes \chi (-z)\big) (\tau \ten \tau) =0.
\end{equation}
Here $\tau$ is  an element  of the Fock space $\mathit{F_{\chi}}$ (in fact, $\tau$ may need to be an element of a certain completion of $\mathit{F_{\chi}}$, as we will discuss in a consequent paper about the solutions to this Hirota equation).

In order to relate this purely algebraic Hirota equation to a system of differential equations we need to bosonize it. As outlined in the introduction, the bosonization  will proceed in 3 stages. The first surprise presented by the CKP case is that,  as we showed in \cite{AngMLB},  there is a second Heisenberg field generated by the field $\chi (z)$ and its descendant field $\chi (-z)$, and therefore two different  bosonizations of  the algebraic Hirota equation are possible:
\begin{prop}\label{prop:Heis-chi}
I. Let $h_\chi^{\mathbb{Z}+1/2}(z)= \frac{1}{2}:\chi (z)\chi(-z):$. We have $h_\chi^{\mathbb{Z}+1/2}(-z)=h_\chi^{\mathbb{Z}+1/2}(z)$, and we index $h_\chi^{\mathbb{Z}+1/2} (z)$ as
$h_\chi^{\mathbb{Z}+1/2} (z)=\sum _{n\in \mathbb{Z}+1/2} h^{\mathbb{Z}+1/2}_{n} z^{-2n-1}$. The field $h_\chi^{\mathbb{Z}+1/2} (z)$ has OPE with itself given by:
\begin{equation}
\label{eqn:HeisOPEsC-t}
 h_\chi^{\mathbb{Z}+1/2}(z)h_\chi^{\mathbb{Z}+1/2} (w)\sim -\frac{z^2 +w^2}{2(z^2 -w^2)^2}\sim -\frac{1}{4}\frac{1}{(z-w)^2} - \frac{1}{4}\frac{1}{(z+w)^2} ,
\end{equation}
and its  modes, $h^{\mathbb{Z}+1/2}_n, \ n\in \mathbb{Z}+1/2$, generate a \textbf{twisted} Heisenberg algebra $\mathcal{H}_{\mathbb{Z}+1/2}$ with relations \\ $[h^{\mathbb{Z}+1/2}_m, h^{\mathbb{Z}+1/2}_n]=-m\delta _{m+n,0}1$, \ $m,n\in \mathbb{Z}+1/2$.\\
II. Let $h_\chi^{\mathbb{Z}}(z)= \frac{1}{4z}\left(:\chi (z)\chi(z):- :\chi (-z)\chi(-z):\right)$. We have $h_\chi^{\mathbb{Z}}(-z)=h_\chi^{\mathbb{Z}}(z)$, and we index $h_\chi^{\mathbb{Z}} (z)$ as
$h_\chi^{\mathbb{Z}} (z)=\sum _{n\in \mathbb{Z}} h^{\mathbb{Z}}_{n} z^{-2n-2}$. The field $h_\chi^{\mathbb{Z}} (z)$ has OPE with itself given by:
\begin{equation}
\label{eqn:HeisOPEsC-ut}
 h_\chi^{\mathbb{Z}}(z)h_\chi^{\mathbb{Z}} (w)\sim -\frac{1}{(z^2 -w^2)^2},
\end{equation}
and its  modes, $h^{\mathbb{Z}}_n, \ n\in \mathbb{Z}$, generate an \textbf{untwisted} Heisenberg algebra $\mathcal{H}_{\mathbb{Z}}$ with relations $[h^{\mathbb{Z}}_m, h^{\mathbb{Z}}_n]=-m\delta _{m+n,0}1$, \ $m,n\in \mathbb{Z}$.
 \end{prop}
The bosonization initiated by the twisted Heisenberg current from the  the above proposition is  studied in \cite{OrlovLeur}. In this paper we study the second bosonization,  initiated by the untwisted Heisenberg current. For simplicity from now on we will denote the untwisted field  $h_\chi^{\mathbb{Z}}(z)$ by $h_\chi(z)$ and its modes by $h_n, \ n\in \mathbb{Z}$.

For the second step in the bosonization process we first need to show that the Heisenberg algebra representation on $\mathit{F_{\chi}}$ is in fact completely reducible.  It is immediate that the representation is a bounded field representation (see e.g. Theorem 3.5 of \cite{Kac}), and we just need to show that $h_0$ is diagonalizable.
 To that effect  we need to introduce various gradings on  $\mathit{F_{\chi}}$. There are at least two types of natural gradings: the first one necessarily derived from the Heisenberg field, specifically from the action of $h_0$, and the second from one of the families of Virasoro fields that we discussed in \cite{AngMLB}.

We first introduce a normal ordered product $:\chi_m \chi_n:$ on the modes $\chi_m$ of the field $\chi (z)$, compatible with the normal ordered product of fields, by:
\begin{equation}
\label{eqn:normord}
:\chi(z)\chi(w): =\sum _{m, n\in \mathbb{Z}+1/2} :\chi_m \chi_n: z^{-m-1/2}w^{-n-1/2}= \sum _{m,n\in \mathbb{Z}} :\chi_{-m-\frac{1}{2}} \chi_{-n-\frac{1}{2}}:z^{m}w^{n},
\end{equation}
and thus for $m, n\in \mathbb{Z}$ this results in the usual "move annihilation operators to the right" approach:
\begin{align*}
:\chi_{-m-\frac{1}{2}}& \chi_{-n-\frac{1}{2}}: =\chi_{-m-\frac{1}{2}}\chi_{-n-\frac{1}{2}} \quad \text{for}\  m+n\neq 1\\
:\chi_{-m-\frac{1}{2}}& \chi_{-n-\frac{1}{2}}: =\chi_{-m-\frac{1}{2}}\chi_{-n-\frac{1}{2}} -(-1)^{m-\frac{1}{2}}=\chi_{-n-\frac{1}{2}}\chi_{-m-\frac{1}{2}} \quad \text{for}\ m+n= -1, n\geq 0,\\
:\chi_{-m-\frac{1}{2}}& \chi_{-n-\frac{1}{2}}: =\chi_{-m-\frac{1}{2}}\chi_{-n-\frac{1}{2}} \quad \text{for}\ m+n= -1, m\geq 0.
\end{align*}
Hence we can express the modes of the field  $h_\chi (z)=\sum _{n\in \mathbb{Z}} h_{n} z^{-2n-2}$ as follows
\begin{equation}
\label{eq:Heismodes}
h_n=\frac{1}{2}\sum _{k\in \mathbb{Z}+1/2}:\chi_{k} \chi_{2n-k}: =\frac{1}{2}\sum _{i\in \mathbb{Z}}:\chi_{-i-\frac{1}{2}} \chi_{2n+ i+\frac{1}{2}}:
\end{equation}
In particular, we have
\begin{equation}
\label{eq:Heismodes}
h_0=\sum _{k\in \mathbb{Z}_{\geq 0}+1/2}:\chi_{-k} \chi_{k}: =\chi_{-\frac{1}{2}}\chi_{\frac{1}{2}} +\chi_{-\frac{3}{2}}\chi_{\frac{3}{2}} +\dots
\end{equation}
Hence it follows that on a monomial $\left(\chi _{-j_k}\right)^{m_k}\dots \left(\chi _{-j_2}\right)^{m_2}\left(\chi _{-j_1}\right)^{m_1}|0\rangle$ in $\mathit{F_{\chi}}$ we have
\begin{equation}
h_0 \Big(\left(\chi _{-j_k}\right)^{m_k}\dots \left(\chi _{-j_2}\right)^{m_2}\left(\chi _{-j_1}\right)^{m_1}|0\rangle \Big)=\Big( \sum_{j_i\in 2\mathbb{Z} +1/2} m_i - \sum_{j_i\in 2\mathbb{Z} -1/2} m_i\Big)\Big(\left(\chi _{-j_k}\right)^{m_k}\dots \left(\chi _{-j_2}\right)^{m_2}\left(\chi _{-j_1}\right)^{m_1}|0\rangle \Big).
\end{equation}
 This shows that $h_0$ is  diagonalizable and thus the Heisenberg algebra representation on $\mathit{F_{\chi}}$ is completely reducible. It also gives $\mathit{F_{\chi}}$ a $\mathbb{Z}$ grading, which we will call \emph{charge} and denote $chg$ (as it is similar to the charge grading in the usual boson-fermion correspondence of type A, i.e., the bosonization related to the KP hierarchy):
 \begin{equation}
chg \big(|0\rangle \big) =0; \quad
chg  \Big(\left(\chi _{-j_k}\right)^{m_k}\dots \left(\chi _{-j_2}\right)^{m_2}\left(\chi _{-j_1}\right)^{m_1}|0\rangle \Big) =\sum_{j_i\in 2\mathbb{Z} +1/2} m_i - \sum_{j_i\in 2\mathbb{Z} -1/2} m_i.
\end{equation}
Example: $chg \big( \chi_{-\frac{1}{2}} |0\rangle \big) =1$; \   $chg \big( \chi_{-\frac{3}{2}} |0\rangle \big) =-1$;\  $chg \big( \chi_{-\frac{3}{2}} \chi_{-\frac{1}{2}} |0\rangle \big) =0$.

Denote the linear span of monomials  of charge $n$ by $\mathit{F^{(n)}_{\chi}}$. The Fock space $\mathit{F_{\chi}}$ has a charge decomposition
\[
\mathit{F_{\chi}} =\oplus_{n\in \mathbb{Z}} \mathit{F^{(n)}_{\chi}}.
\]
In the usual boson-fermion correspondence (of type A), the charge decomposition is in fact the decomposition of the Fock space  in terms of irreducible Heisenberg modules (see e.g. Theorem 5.1 of \cite{Kac}, as well as the more detailed descriptions in \cite{KacRaina}, \cite{Miwa-book}); i.e., each charge component $ \mathit{F^{(n)}_{\chi}}$ is in fact a Heisenberg irreducible module.
This is not the case here: for example, the vector
\[
v_{4;  0} :=\chi_{-\frac{3}{2}}^2 \chi_{-\frac{1}{2}}^2|0\rangle  -2\chi_{-\frac{7}{2}}\chi_{-\frac{1}{2}} |0\rangle +2\chi_{-\frac{5}{2}}\chi_{-\frac{5}{2}}|0\rangle
\]
is of charge 0, but  we can also directly check that $h_n v_{4;  0} =0$ for any $n>0$. Thus $v_{4; 0}$ is another highest weight vector of charge 0 for the action of the Heisenberg algebra, besides the vacuum $|0\rangle $. Therefore the charge 0 component $\mathit{F^{(0)}_{\chi}}$ is not irreducible as a Heisenberg module, in contrast to the usual boson-fermion correspondence (of type A), Similarly, neither are the other charge components, as we shall see.

Next, there is  a   $\frac{1}{2}\mathbb{Z}$  grading on $\mathit{F_{\chi}}$ we will call \emph{degree} and denote by $deg$, which we obtain by using one of the three families of Virasoro fields that were discussed in \cite{AngMLB}. In \cite{AngMLB} we introduced the  descendent fields
$\beta_\chi (z^2), \gamma_\chi (z^2)$  defined by
\begin{equation}
\beta_\chi (z^2) =\frac{\chi(z) -\chi (-z)}{2z}; \quad \gamma_\chi (z^2) =\frac{\chi(z) +\chi (-z)}{2}.
\end{equation}
These fields have  OPEs:
\begin{align}
\beta_\chi(z^2)\beta_\chi(w^2)\sim 0; \quad \gamma_\chi(z^2)\gamma_\chi(w^2)\sim 0; \quad
\beta_\chi(z^2)\gamma_\chi(w^2)\sim \frac{1}{z^2-w^2}; \quad \gamma_\chi(z^2)\beta_\chi(w^2)\sim -\frac{1}{z^2-w^2}.
\end{align}
In particular, we have
\begin{align}
\beta_\chi (z^2) =\sum_{m\in \mathbb{Z}} \chi_{-2m+\frac{1}{2}}(z^2)^{m-1}; \quad \gamma_\chi (z^2) =\sum_{m\in \mathbb{Z}} \chi_{-2m-\frac{1}{2}}(z^2)^{m}
\end{align}
Hence we can translate the following Virasoro field  (\cite{AngMLB}) from the $\beta-\gamma$ system
\begin{equation}
L_2^{\beta\gamma;\  (\lambda, \mu)} (z)=\lambda :\left(\partial_z\beta (z)\right)\gamma (z): +(\lambda +1):\beta (z)\left(\partial_z\gamma (z)\right): +\frac{\mu}{z} :\beta (z)\gamma (z): +\frac{(2\lambda +1)\mu -\mu^2}{2z^2},
\end{equation}
into a Virasoro action on $\mathit{F_{\chi}}$. For simplicity we will consider only the case $\mu =0$, and we have
\begin{equation}
\label{eqn:Vir2}
L^\lambda (z^2)= \sum_{n\in \mathbb{Z}}L_n (z^2)^{-n-2} = -\sum_{n\in \mathbb{Z}}\Big(\sum_{k+l=n}  \left(\lambda (k+1) +(\lambda +1)l\right) :\chi_{2k+\frac{1}{2}}\chi_{2l-\frac{1}{2}}:\Big)(z^2)^{-n-2},
\end{equation}
in particular
\begin{equation}
\label{eqn:VirGradingLambda}
L^\lambda_0 = -\sum_{k\in \mathbb{Z}}\left(\lambda  +k\right) :\chi_{-2k+\frac{1}{2}}\chi_{2k-\frac{1}{2}}:.
\end{equation}
We can further vary $\lambda$ ($\lambda =-\frac{1}{2}$ is usually chosen in conformal field theory), but a useful
choice here is $\lambda =-\frac{1}{4}$. In that case, we have
\begin{equation}
\label{eqn:VirGrading}
L_0 = \frac{1}{2}\Big( \frac{1}{2}:\chi_{-\frac{1}{2}}\chi_{\frac{1}{2}}: -\frac{3}{2}:\chi_{-\frac{3}{2}}\chi_{\frac{3}{2}}: +\frac{5}{2}:\chi_{-\frac{5}{2}}\chi_{\frac{5}{2}}: -\dots \Big).
\end{equation}
Hence
\begin{equation}
L_0 \Big(\left(\chi _{-j_k}\right)^{m_k}\dots \left(\chi _{-j_2}\right)^{m_2}\left(\chi _{-j_1}\right)^{m_1}|0\rangle \Big)= \frac{1}{2}\left(m_k\cdot j_k +\dots m_2\cdot j_2 +m_1\cdot j_1\right)\Big(\left(\chi _{-j_k}\right)^{m_k}\dots \left(\chi _{-j_2}\right)^{m_2}\left(\chi _{-j_1}\right)^{m_1}|0\rangle \Big).
\end{equation}
Discarding the factor of $\frac{1}{2}$ we have the  $deg$ grading on $\mathit{F_{\chi}}$ (also used in
\cite{OrlovLeur}):
 \begin{equation}
deg \big(|0\rangle \big) =0; \quad deg \chi _{-j} |0\rangle =j; \quad
  deg \Big(\left(\chi _{-j_k}\right)^{m_k}\dots \left(\chi _{-j_2}\right)^{m_2}\left(\chi _{-j_1}\right)^{m_1}|0\rangle \Big) =m_k\cdot j_k +\dots m_2\cdot j_2 +m_1\cdot j_1,
\end{equation}
where $j_k>\dots > j_2 >j_1 > 0, \ j_i\in \mathbb{Z}+\frac{1}{2}$,\  $m_i > 0, m_i\in \mathbb{Z}, \ i=1, 2, \dots, k$. Consequently we have a degree decomposition of $\mathit{F_{\chi}}$ as in \cite{OrlovLeur}, which now we know is actually derived from the Virasoro operator component $L_0$. The formal character is given by:
\begin{equation}
dim_q \mathit{F_{\chi}}:= tr_{\mathit{F_{\chi}}}q^{2L_0} =\sum_{k\in \frac{1}{2}\mathbb{Z}} dim \Big(span \{ v\in \mathit{F_{\chi}}\ | \ deg (v) =k \}\Big) q^k.
\end{equation}
We can form also the character with respect to both the $L_0$ and $h_0$ grading operators (they are both diagonalizable):
\begin{equation}
dim_{q,z} \mathit{F_{\chi}}:= tr_{\mathit{F_{\chi}}}q^{2L_0}z^{h_0}.
\end{equation}
Now observing that acting by $\chi _{-2j-\frac{1}{2}}$, $\j\geq 0$, on a monomial $\left(\chi _{-j_k}\right)^{m_k}\dots \left(\chi _{-j_2}\right)^{m_2}\left(\chi _{-j_1}\right)^{m_1}|0\rangle$   will produce a factor of $z^{+1}q^{2j+ \frac{1}{2}}$, and acting by $\chi _{-2j+\frac{1}{2}}$,  $\j\geq 1$ will produce a factor of $z^{-1}q^{2j- \frac{1}{2}}$, it is immediate that
\begin{equation}
\label{eqn:dimzq}
dim_{q,z} \mathit{F_{\chi}}= \frac{1}{\prod_{j\in \mathbb{Z}_{+}} \big(1-zq^{2j- \frac{3}{2}}\big) \big(1-z^{-1}q^{2j- \frac{1}{2}}\big)}.
\end{equation}
The formula
\begin{equation}\label{eqn:gradeddim}
dim_q \mathit{F_{\chi}} =\frac{1}{\prod_{j\in \mathbb{Z}_{+}} \big(1-q^{\frac{2j-1}{2}}\big)}
\end{equation}
of \cite{OrlovLeur} then follows from setting $z=1$ in \eqref{eqn:dimzq}.

\begin{lem} The following relations hold:
\[
[ L_0, h_{m} ] =-m h_m, \quad \forall \ m\in \mathbb{Z},
\]
and thus for any $v\in \mathit{F_{\chi}}$ we have
\begin{equation}
\label{eqn:degreechange}
deg (h_{-m}v) = 2m +deg (v) , \quad \forall \ m\in \mathbb{Z}_+.
\end{equation}
\end{lem}
\begin{proof} By using the relation with the $\beta\gamma$ system we can calculate the OPE between \\ $L (z^2)=-\frac{1}{4} :\left(\partial_{z^2}\beta (z^2)\right)\gamma (z^2): +\frac{3}{4}:\beta (z^2)\left(\partial_{z^2}\gamma (z^2)\right):$ and $h_\chi(z) =:\beta (z^2)\gamma (z^2):$ via Wick's Theorem.
The calculations are straightforward.\end{proof}

Since the conditions of Theorem 3.5 of \cite{Kac} are satisfied, the Heisenberg module $\mathit{F_{\chi}}$ is completely reducible, and is a direct sum of irreducible highest weight Heisenberg modules, each isomorphic to
\[
\mathbb{C}[h_{-1}, h_{-2}, \dots , h_{-n}, \dots ] \cdot v
\]
for some highest weight vector $v$, for which $h_nv =0$ for any $n>0$. It is a
 well known fact (see e.g. \cite{KacRaina}, \cite{FLM}) that any irreducible highest weight module of the Heisenberg
algebra $\mathcal{H}_{\mathbb{Z}}$ introduced in Section 2 is isomorphic to the polynomial algebra with infinitely many variables $\mathit{B_\lambda}\cong \mathbb{C}[x_1,
x_2, \dots , x_n, \dots ]$ where $v\mapsto 1$ and:
\begin{equation}
h_n\mapsto i\partial _{x_{n}}, \quad h_{-n} \mapsto
inx_n\cdot, \quad \text{for any} \ \ n\in \mathbb{N}, \quad h_0\mapsto \lambda\cdot , \ \lambda \in \mathbb{C}.
\end{equation}
In fact we can introduce an arbitrary re-scaling $s_n\neq 0, \ s_n\in \mathbb{C}$, for $n\neq 0$ only, so that
\begin{equation}
h_n\mapsto is_n\partial _{x_{n}}, \quad h_{-n} \mapsto
is_n^{-1}nx_n\cdot, \quad \text{for any} \ \ n\in \mathbb{N}, \quad h_0\mapsto \lambda\cdot.
\end{equation}
Thus each of the irreducible modules in our Heisenberg decomposition is isomorphic to   $\mathit{B_\lambda}\cong \mathbb{C}[x_1,
x_2, \dots , x_n, \dots ]$ for some $\lambda \in \mathbb{C}$ determined by the charge of the highest weight vector generating the module.
Now if $v$ is a highest weight vector, which induces an irreducible module $V =\mathbb{C}[h_{-1}, h_{-2}, \dots , h_{-n}, \dots ] \cdot v\cong \mathit{B_\lambda}$, then  as a consequence of \eqref{eqn:degreechange} $V$ has graded dimension
\[
dim_q V =\frac{q^{deg (v)}}{\prod_{i\in \mathbb{Z}_{+}} (1-q^{2i})}.
\]
Since $\mathit{F_{\chi}}$ is a direct sum of such irreducible modules, we have
\[
dim_q \mathit{F_{\chi}} =\frac{\sum_{\mathfrak{p} \in \mathfrak{P}_{tdo}} q^{deg (v_{\mathfrak{p})}}}{\prod_{i\in \mathbb{Z}_{+}} (1-q^{2i})},
\]
where the summation is over an as yet unknown indexing set $\mathcal{P}_{tdo}$. By comparing this formula for the graded dimension with \eqref{eqn:gradeddim}, we have
\[
\sum_{\mathfrak{p}\in \mathfrak{P}_{tdo}} q^{deg (v_{\mathfrak{p})}} =\frac{\prod_{i\in \mathbb{Z}_{+}} (1-q^{2i})}{\prod_{i\in \mathbb{Z}_{+}} (1-q^{\frac{2i-1}{2}})} = \frac{\prod_{i\in \mathbb{Z}_{+}} (1-q^{2i})\prod_{i\in \mathbb{Z}_{+}} (1+q^{\frac{2i-1}{2}})}{\prod_{i\in \mathbb{Z}_{+}} (1-q^{2i-1})}
\]
Now using the Jacobi triple identity in one of its forms:
\begin{equation}
\label{eqn:Jacid}
\prod_{i=1}^{\infty} (1-q^{i})(1+zq^{i-1})(1+z^{-1}q^{i})=\sum_{m\in \mathbb{Z}}z^m q^{\frac{m(m-1)}{2}}.
\end{equation}
we have, by  setting $z=1$,
\[
2\sum_{m\in \mathbb{Z}_{\geq 0}}q^{T_m} =2\prod_{i=1}^{\infty} (1-q^{2i})(1+q^{i})= 2\frac{\prod_{i=1}^{\infty} (1-q^{2i})(1-q^{2i-1})(1+q^{i})}{\prod_{i=1}^{\infty} (1-q^{2i-1})} =2 \frac{\prod_{i=1}^{\infty} (1-q^{i})(1+q^{i})}{\prod_{i=1}^{\infty} (1-q^{2i-1})}
\]
where $T_m$ denotes the $m$-th triangular number---  $T_m:=1+2+\dots +m =\frac{m(m+1)}{2}$, with $T_0 =0$. Hence by necessity we re-derived a known\footnote{We couldn't find a reference for this formula.} formula for the triangular numbers:
\[
\sum_{m\in \mathbb{Z}_{\geq 0}}q^{T_m} =1+q +q^3 +q^6 +q^{10} +\dots + q^{T_m} +\dots = \frac{\prod_{i=1}^{\infty} (1-q^{2i})}{\prod_{i=1}^{\infty} (1-q^{2i-1})}.
\]
Using this formula, we have
\begin{equation}
\sum_{\mathfrak{p}\in \mathfrak{P}_{tdo}} q^{deg (v_{\mathfrak{p})}} = \Big(\sum_{m\in \mathbb{Z}_{\geq 0}}q^{T_m}\Big)\cdot \prod_{i\in \mathbb{Z}_{+}} (1+q^{\frac{2i-1}{2}}).
\end{equation}
Since the right-hand side is now  a sum with positive coefficients, it determines the indexing set $\mathfrak{P}_{tdo}$, namely it consists of distinct partitions of the type
\begin{equation}
\mathfrak{P}_{tdo}=\{\mathfrak{p} =(T_m, \lambda_1, \lambda_2, \dots , \lambda_k) \ |\ T_m-\text{triangular\ number}, \ \lambda_1>\lambda_2>\dots >\lambda_k, \lambda_i \in \frac{1}{2}+\mathbb{Z}_{\geq 0}, \ i=1, \dots , k \}.
\end{equation}
As usual, the weight $| \mathfrak{p} | $   of such a partition $\mathfrak{p}$ is the sum of its parts, $| \mathfrak{p} |: =T_m +  \lambda_1 + \lambda_2 + \dots + \lambda_k$.

Hence we arrive at  the following proposition, which provides the decomposition of $\mathit{F_{\chi}}$ into irreducible Heisenberg modules, thus completing the second step in the process of bosonization:
\begin{prop} \label{prop:HeisDecomp}
For the action of the Heisenberg algebra $\mathcal{H}_{\mathbb{Z}}$ on $\mathit{F_{\chi}}$,  the number of highest weight vectors  of degree $n\in \frac{1}{2}\mathbb{Z}$  equals the number of partitions $\mathfrak{p}\in \mathfrak{P}_{tdo}$ of weight $n$. Thus as Heisenberg modules
\begin{equation}
\mathit{F_{\chi}}\cong  \oplus_{\mathfrak{p}\in \mathfrak{P}_{tdo}}   \mathbb{C}[x_1, x_2, \dots , x_n, \dots ].
\end{equation}
\end{prop}
\begin{example}
We can calculate the highest weight vectors of given  degree by brute force. For the fist few degrees we have
\begin{align*}
\sum_{\mathfrak{p}\in \mathfrak{P}_{tdo}} q^{deg (v_{\mathfrak{p})}} &= \Big(\sum_{m\in \mathbb{Z}_{\geq 0}}q^{T_m}\Big)\cdot \prod_{i\in \mathbb{Z}_{+}} (1+q^{\frac{2i-1}{2}})\\
& =\left(1+q+q^3 +q^6 +\dots\right)\left(1+q^{\frac{1}{2}} + q^{\frac{3}{2}} +q^2 +q^{\frac{5}{2}} +q^3 +q^{\frac{7}{2}} +2q^4 +2q^{\frac{9}{2}} +\dots\right)\\
& =1+ q^{\frac{1}{2}} + q + 2q^{\frac{3}{2}} + q^2 + 2q^{\frac{5}{2}} +3q^3 +3q^{\frac{7}{2}} +3q^4 +4q^{\frac{9}{2}} +\dots
\end{align*}
The corresponding highest weight vectors are (in each degree the maximum charge of the highest weight vectors starts at twice that degree,  and also the charges inside each degree are equivalent modulo 4):
\begin{displaymath}
\begin{array}{c|c}
1 & |0\rangle \hspace{\stretch{1}}\\
q^{\frac{1}{2}} & \chi_{-\frac{1}{2}}|0\rangle  \hspace{\stretch{1}}\\
q & \chi_{-\frac{1}{2}}^2|0\rangle  \hspace{\stretch{1}}\\
2q^{\frac{3}{2}} & \chi_{-\frac{1}{2}}^3|0\rangle ;\ \  \chi_{-\frac{3}{2}}|0\rangle  \hspace{\stretch{1}}\\
q^2 & \chi_{-\frac{1}{2}}^4|0\rangle   \hspace{\stretch{1}}\\
2q^{\frac{5}{2}} & \chi_{-\frac{1}{2}}^5|0\rangle ; \ \ \chi_{-\frac{3}{2}} \chi_{-\frac{1}{2}}^2|0\rangle  +2\chi_{-\frac{5}{2}}|0\rangle  \hspace{\stretch{1}}\\
3q^3 & \chi_{-\frac{1}{2}}^6|0\rangle ;\ \ \chi_{-\frac{3}{2}}\chi_{-\frac{1}{2}}^3|0\rangle +3 \chi_{-\frac{5}{2}}\chi_{-\frac{1}{2}}|0\rangle; \ \  \chi_{-\frac{3}{2}}^2|0\rangle  \hspace{\stretch{1}}\\
3q^{\frac{7}{2}} & \chi_{-\frac{1}{2}}^7|0\rangle ;  \ \ \chi_{-\frac{3}{2}}\chi_{-\frac{1}{2}}^4|0\rangle +4 \chi_{-\frac{5}{2}}\chi_{-\frac{1}{2}}^2|0\rangle; \ \ \chi_{-\frac{3}{2}}^2 \chi_{-\frac{1}{2}}|0\rangle  -2\chi_{-\frac{7}{2}}|0\rangle   \hspace{\stretch{1}}\\
3q^4 & \chi_{-\frac{1}{2}}^8|0\rangle ; \ \ \chi_{-\frac{3}{2}} \chi_{-\frac{1}{2}}^5|0\rangle  +5\chi_{-\frac{5}{2}}\chi_{-\frac{1}{2}}^3|0\rangle ; \ \ \chi_{-\frac{3}{2}}^2\chi_{-\frac{1}{2}}^2|0\rangle -2 \chi_{-\frac{7}{2}}\chi_{-\frac{1}{2}}|0\rangle +2 \chi_{-\frac{5}{2}}\chi_{-\frac{3}{2}}|0\rangle  \hspace{\stretch{1}}\\
4q^{\frac{9}{2}} & \chi_{-\frac{1}{2}}^9|0\rangle ; \ \ \chi_{-\frac{3}{2}} \chi_{-\frac{1}{2}}^6|0\rangle  +6\chi_{-\frac{5}{2}}\chi_{-\frac{1}{2}}^4|0\rangle ; \ \ \chi_{-\frac{3}{2}}^2\chi_{-\frac{1}{2}}^3|0\rangle +6 \chi_{-\frac{9}{2}}|0\rangle +6 \chi_{-\frac{5}{2}}\chi_{-\frac{3}{2}}\chi_{-\frac{1}{2}}|0\rangle  ;\ \ \chi_{-\frac{3}{2}}^3|0\rangle \hspace{\stretch{1}}
 \end{array}
 \end{displaymath}
 \end{example}
There are several families of highest weight vectors, for instance one can easily check that
$ \chi_{-\frac{1}{2}}^n|0\rangle$, \ $\chi_{-\frac{3}{2}}^n|0\rangle$ and $ \chi_{-\frac{3}{2}}\chi_{-\frac{1}{2}}^{n+2}|0\rangle +(n+2) \chi_{-\frac{5}{2}}\chi_{-\frac{1}{2}}^{n}|0\rangle$ are highest weight vectors for any $n\geq 0$. Also, observe that at any given weight  $\frac{n}{2}$ for $n\in \mathbb{Z}$,
$\chi_{-\frac{1}{2}}^{n}|0\rangle$ is the highest weight vector of the highest  charge ($n$) with that degree.
\begin{remark} It would be interesting to derive a formula giving a correspondence between a partition $\mathfrak{p}\in \mathfrak{P}_{tdo}$ of weight $n$ and the highest weight vector corresponding to that partition, or even the charge of that highest weight vector. As the weights of the partitions grow, the charges are less straightforward to calculate.  For example at weight $\frac{13}{2}$ there are 7 partitions from $\mathfrak{P}_{tdo}$ and one can calculate by brute force that there is a highest weight vector of charge 13, a highest weight vector of charge 9, two highest weight vectors of charge 5,  two highest weight vectors of charge 1 and  a highest weight vector of charge $-3$.
\end{remark}

Denote by \textbf{$\mathit{F^{hwv}_\chi}$} the vector space spanned by all the highest weight vectors for the Heisenberg action.  To accomplish the third step in the bosonization process,  in the next section we will first  show that $\mathit{F^{hwv}_\chi}$ has  a structure realizing  the  symplectic fermion  super vertex algebra.

\section{Symplectic fermions: vertex algebra structure on the space spanned by the  Heisenberg highest weight vectors}
\label{HighestWeightVectors}

As usual, for a rational function $f(z,w)$,  with poles only at $z=0$,  $z=\pm w$, we denote by $i_{z,w}f(z,w)$
the expansion of $f(z,w)$ in the region $\abs{z}\gg \abs{w}$ (the region in the complex $z$ plane outside   the points $z=0, \pm w$), and correspondingly for
$i_{w,z}f(z,w)$.
\begin{lem}\label{lem:Heisbeta} The following OPEs  hold:
\begin{equation}
h_\chi(z)\beta_\chi(w^2) \sim -\frac{1}{z^2-w^2} \beta_\chi(w^2); \quad h_\chi(z)\gamma_\chi(w^2)\sim \frac{1}{z^2-w^2} \gamma_\chi(w^2)
\end{equation}
\end{lem}
\begin{proof} By direct application of Wick's Theorem. \end{proof}
Denote
\begin{equation}
V^-(z) = \exp \Big(-\sum_{n>0}\frac{1}{n}h_n z^{-2n}\Big);\quad \quad
V^+(z) = \exp \Big(\sum_{n>0}\frac{1}{n}h_{-n} z^{2n}\Big)
\end{equation}
Consequently, we will write
\begin{equation}
V^-(z)^{-1} = \exp \Big(\sum_{n>0}\frac{1}{n} h_n z^{-2n}\Big);\quad \quad
V^+(z)^{-1} = \exp \Big(-\sum_{n>0}\frac{1}{n} h_{-n} z^{2n}\Big)
\end{equation}
\begin{lem} The following commutation relations hold:
\label{lem:comrelV}
\begin{align}
&[h_\chi(z), V^-(w) ]  =  -\Big(\sum_{n>0}\frac{z^{2n-2}}{w^{2n}}\Big) V^-(w) = -i_{w, z}\frac{1}{w^2 -z^2}V^-(w) =i_{w, z}\frac{1}{z^2 -w^2}V^-(w) ;\\
&[h_\chi(z), V^+(w) ]  = -\Big(\sum_{n>0}\frac{w^{2n}}{z^{2n+2}}\Big) V^+(w) =-i_{z, w}\frac{w^2}{z^2(z^2 -w^2)}V^+(w); \\
&V^-(z) V^+(w) = i_{z, w}\frac{z^2}{z^2 -w^2}V^+(w)V^-(z); \\
&V^-(z)^{-1} V^+(w)^{-1} = i_{z, w}\frac{z^2}{z^2 -w^2}V^+(w)^{-1}V^-(z)^{-1}; \\
&V^-(z)^{-1} V^+(w) = \frac{z^2 -w^2}{z^2}V^+(w)V^-(z)^{-1}; \\
&V^-(z) V^+(w)^{-1} = \frac{z^2 -w^2}{z^2}V^+(w)^{-1}V^-(z).
\end{align}
\end{lem}
\begin{proof} The proof is by direct calculation on the first two relations. On the other four we apply the Baker-Campbell-Hausdorff formula, we will only show it for one of the relations:
\begin{align*}
V^-(z) V^+(w) &=\exp\left(-[\sum_{m>0}\frac{1}{m} h_m z^{-2m}, \sum_{n>0}\frac{1}{n} h_{-n} w^{2n}]\right)\cdot V^+(w)V^-(z)\\
&=\exp \left(\sum_{m>0} \frac{1}{m}\frac{w^{2m}}{z^{2m}}\right)\cdot V^+(w)V^-(z)=i_{z, w}\exp \left(-\ln \left(1-\frac{w^2}{z^2}\right)\right) \cdot V^+(w)V^-(z).
\end{align*}
\end{proof}
Observe that $V^-(z)$ and $V^+(z)$ are actually functions of $z^2$. With that in mind, denote
\begin{equation}
H^{\beta} (z^2) = V^+(z)^{-1}\beta_\chi(z^2)z^{-2h_0}V^-(z)^{-1}, \quad \quad H^{\gamma} (z^2) = V^+(z)\gamma_\chi (z^2)z^{2h_0}V^-(z).
\end{equation}
\begin{remark}
We want to mention that in the case of the usual boson-fermion correspondence (for the KP hierarchy, aka of type A), one introduces an invertible operator $u$ from the subspace of  charge $m$ to the subspace of charge $m+1$ (see e.g. \cite{Kac}, Section 5.2), mapping the unique--in that case---highest weight vector of charge $m$ to the unique highest weight vector of charge $m+1$. That operator $u$ is then used to define the (simpler) counterparts of $H^{\beta} (z^2)$ and $H^{\gamma} (z^2)$.  As we saw in the previous section, in our case the charge components $\mathit{F^{(m)}_{\chi}}$ are  not irreducible, and therefore such an invertible operator $u$ doesn't exist, at least not as an invertible operator sending a highest weight vector to highest weight vector.
\end{remark}
\begin{lem}
\label{lem:H-betawithHeis}
The following commutation relations hold:
\begin{equation}
[h_\chi(z), H^{\beta} (w^2)] =-\frac{1}{z^2}H^{\beta} (w^2), \quad \quad [h_\chi(z), H^{\gamma} (w^2)] =\frac{1}{z^2}H^{\gamma} (w^2).
\end{equation}
Therefore $H^{\beta} (z^2)$ and $H^{\gamma} (z^2)$ can be considered  fields (vertex operators) on $\mathit{F^{hwv}_\chi}$, i.e., for each $v\in \mathit{F^{hwv}_\chi}$, we have $H^{\beta} (z^2) v\in \mathit{F^{hwv}_\chi}((z^2))$ and $H^{\gamma} (z^2) v\in \mathit{F^{hwv}_\chi}((z^2))$.
\end{lem}
\begin{proof}
We have
\begin{align*}
[h_\chi(z), H^{\beta} (w^2)]& = [h_\chi(z), V^+(w)^{-1}\beta (w^2)w^{-2h_0}V^-(w)^{-1}]\\
& =\Big(\sum_{n>0}\frac{z^{2n-2}}{w^{2n}}-i_{z, w}\frac{1}{z^2-w^2} -i_{w, z}\frac{1}{w^2-z^2}+\sum_{n>0}\frac{w^{2n}}{z^{2n+2}}\Big)H^{\beta} (w^2)=-\frac{1}{z^2}H^{\beta} (w^2).
\end{align*}
Hence we see that
\begin{equation}
\label{eqn:Hbetacharge}
[h_n, H^{\beta} (w^2)]=0, \quad \text{for\ any}\quad n\neq 0; \quad [h_0, H^{\beta} (w^2)]= - H^{\beta} (w^2).
\end{equation}
We can similarly see that
\begin{align*}
[h_\chi(z), H^{\gamma} (w^2)]& = [h_\chi(z), V^+(w)\gamma (w^2)w^{2h_0}V^-(w)]\\
& =\Big(-\sum_{n>0}\frac{z^{2n-2}}{w^{2n}}+i_{z, w}\frac{1}{z^2-w^2} +i_{w, z}\frac{1}{w^2-z^2}-\sum_{n>0}\frac{w^{2n}}{z^{2n+2}}\Big)H^{\gamma} (w^2)=\frac{1}{z^2}H^{\gamma} (w^2).
\end{align*}
Thus
\begin{equation}
\label{eqn:Hgammacharge}
[h_n, H^{\gamma} (w^2)]=0, \quad \text{for\ any}\quad n\neq 0; \quad [h_0, H^{\gamma} (w^2)]= H^{\gamma} (w^2).
\end{equation}
Now let $v$ be a highest weight vector, i.e., $v\in \mathit{F^{hwv}_\chi}$; from \eqref{eqn:Hbetacharge}  it is clear that
\[
h_n H^{\beta} (z^2)v =H^{\beta} (z^2)h_n v =0, \quad  \text{for\ any}\quad n> 0.
\]
Hence the coefficients of $H^{\beta} (z^2) v$  are in fact highest weight vectors themselves, i.e.,  $H^{\beta} (z^2) v\in \mathit{F^{hwv}_\chi}((z^2))$ (instead of the more general $\mathit{F_\chi}((z))$). Therefore we can view the field $H^{\beta} (z^2)$ as a field on $\mathit{F^{hwv}_\chi}$, instead of more generally on $\mathit{F_\chi}$. Similarly for $H^{\gamma} (z^2)$.
\end{proof}
As mentioned above, in the case of the boson-fermion correspondence of type A (the bosonization of the KP hierarchy), the counterparts of the fields $H^{\beta} (z^2)$ and $H^{\gamma} (z^2)$ are the simple operators  $u^{-1}$ and $u$, see e.g. \cite{Kac}, Section 5.2  (which can be identified with $e^{-\alpha}$ and $e^{\alpha}$ if one identifies the vector space of highest weight vectors in that case with $\mathbb{C}[e^{\alpha}, e^{-\alpha}]$). In particular there the operators $u^{-1}$ and $u$ are actually independent of $z$. This is not the case for the fields $H^{\beta} (z^2)$ and $H^{\gamma} (z^2)$, as we will show:
\begin{prop}\label{prop:symOPEs}
The following commutation relations hold:
\begin{align}
&\{H^{\beta} (z^2), H^{\gamma} (w^2)\} =i_{z, w}\frac{1}{(z^2 -w^2)^2}- i_{w, z}\frac{1}{(w^2 -z^2)^2}; \\
&\{H^{\beta} (z^2), H^{\beta} (w^2)\} = 0; \quad \{H^{\gamma} (z^2), H^{\gamma} (w^2)\} = 0.
\end{align}
Here we use the notation $\{A, B\}: =AB +BA$ for two operators $A, B$.
\end{prop}
If we use the delta function notation (see \cite{Kac}),
\[
\delta(z, w): =\sum _{n\in\mathbb{Z}} \frac{z^n}{w^{n+1}} = i_{z, w}\frac{1}{z -w} +i_{w, z}\frac{1}{w -z},
\]
the nontrivial commutation relation in the proposition above can be written as
\[
\{H^{\beta} (z^2), H^{\gamma} (w^2)\} =\partial_{w^2}\delta(z^2, w^2); \quad \{H^{\gamma} (z^2), H^{\beta} (w^2)\} =\partial_{z^2}\delta(z^2, w^2).
\]
For the proof of this proposition we need the following
\begin{lem} The following commutation relations hold:
\label{lem:comrelbetaV}
\begin{align}
&\beta_\chi(z^2) V^+(w) = i_{z, w}\frac{z^2}{z^2 -w^2}V^+(w)\beta_\chi(z^2), \quad \beta_\chi(z^2) V^+(w)^{-1} = \frac{z^2 -w^2}{z^2}V^+(w)^{-1}\beta_\chi(z^2); \\
&\beta_\chi(z^2) V^-(w) = \frac{w^2 -z^2}{w^2}V^-(w)\beta_\chi(z^2), \quad \beta_\chi(z^2) V^-(w)^{-1} = i_{w, z} \frac{w^2}{w^2 -z^2}{w^2}V^-(w)^{-1}\beta_\chi(z^2); \\
&\gamma_\chi(z^2) V^+(w) = \frac{z^2 -w^2}{z^2}V^+(w)\gamma_\chi(z^2), \quad \gamma_\chi(z^2) V^+(w)^{-1} = i_{z, w}\frac{z^2}{z^2 -w^2}V^+(w)^{-1}\gamma_\chi(z^2);\\
&\gamma_\chi(z^2) V^-(w) = \frac{w^2 -z^2}{w^2}V^-(w)\gamma_\chi(z^2), \quad \gamma_\chi(z^2) V^-(w)^{-1} = i_{w, z}\frac{w^2}{w^2 -z^2} V^-(w)^{-1}\gamma_\chi(z^2).
\end{align}
\end{lem}
\begin{proof}
From the definition of $H^{\beta} (z^2)$ we have
 \begin{align*}
\beta_\chi(z^2) V^+(w) &= V^+(z)H^{\beta} (z^2)z^{2h_0}V^-(z) V^+(w) = V^+(z)H^{\beta} (z^2)z^{2h_0} i_{z, w}\frac{z^2}{z^2 -w^2}V^+(w)V^-(z)\\
&= i_{z, w}\frac{z^2}{z^2 -w^2}V^+(w) V^+(z)H^{\beta} (z^2)z^{2h_0}V^-(z) = i_{z, w}\frac{z^2}{z^2 -w^2}V^+(w)\beta_\chi(z^2).
\end{align*}
Here we used both Lemma \ref{lem:H-betawithHeis}, namely that $H^{\beta} (z^2)$ commutes with both $V^+(w)$ and $V^-(z)$, as well as Lemma \ref{lem:comrelV}.  Similarly
 \begin{align*}
\gamma_\chi(z^2) V^-(w) &= V^+(z)^{-1}H^{\gamma} (z^2)z^{-2h_0}V^-(z)^{-1} V^-(w) = V^+(z)^{-1}V^-(w) H^{\gamma} (z^2)z^{-2h_0}V^-(z)^{-1} \\
&= \left(1-\frac{z^2}{w^2}\right) V^-(w) V^+(z)^{-1} H^{\gamma} (z^2)z^{-2h_0}V^-(z)^{-1}  = \frac{w^2 -z^2} {w^2} V^-(w)\gamma_\chi(z^2).
\end{align*}
The other relations are proved similarly.
\end{proof}
We now return to the proof of the Proposition.
\begin{proof}
We will prove the first of the nontrivial relations, the other is proved similarly. We use the commutation relations from Lemma \ref{lem:comrelV}, and commute successively the annihilating $V^-(z)^{-1}$ to the right, and the creating $V^+(w)$ to the left:
\begin{align*}
H^{\beta} (z^2) H^{\gamma} (w^2)& =  V^+(z)^{-1}\beta_\chi (z^2)z^{-2h_0}V^-(z)^{-1} V^+(w)\gamma_\chi (w^2)w^{2h_0}V^-(w)\\
&= \frac{z^2 -w^2}{z^2} V^+(z)^{-1}\beta_\chi (z^2)z^{-2h_0}V^+(w)V^-(z)^{-1}\gamma_\chi (w^2)w^{2h_0}V^-(w)\\
&= \frac{z^2 -w^2}{z^2}\cdot i_{z, w}\frac{z^2}{z^2 -w^2} V^+(z)^{-1}V^+(w)\beta_\chi (z^2)z^{-2h_0}V^-(z)^{-1}\gamma_\chi (w^2)w^{2h_0}V^-(w)\\
&= \frac{z^2 -w^2}{z^2}\cdot i_{z, w}\frac{z^2}{z^2 -w^2}\cdot i_{z, w}\frac{z^2}{z^2 -w^2} V^+(z)^{-1}V^+(w)\beta_\chi (z^2)z^{-2h_0}\gamma_\chi (w^2)w^{2h_0}V^-(z)^{-1}V^-(w)\\
&= i_{z, w}\frac{z^2}{z^2 -w^2} V^+(z)^{-1}V^+(w)\beta_\chi (z^2)z^{-2h_0}\gamma_\chi (w^2)w^{2h_0}V^-(z)^{-1}V^-(w).
\end{align*}
Now we need to interchange $z^{-2h_0}$ and $\gamma_\chi (w^2)$. From Lemma \ref{lem:Heisbeta} we have $h_0 \gamma_\chi (w^2) =\gamma_\chi (w^2)(h_0 +1)$, or we can see directly from
\[
\gamma_\chi (w^2) = \frac{\chi(w) +\chi (-w)}{2} = \sum _{n\in \mathbb{Z}} \chi _{-2n-1/2} w^{2n} = \dots +\chi _{3/2}w^{-2}+ \chi _{-1/2} +\chi _{-5/2}w^2 +\dots ,
\]
in addition to  the fact that acting by $ \chi _{-2n-1/2}$ adds charge of 1, that
\[
z^{-2h_0}\gamma_\chi (w^2)=\frac{1}{z^2}\gamma_\chi (w^2)z^{-2h_0}.
\]
Finally we have from the OPE  of $\beta_\chi (z^2)$ with $\gamma_\chi (w^2)$, plus the definition of a normal ordered product that
\[
\beta_\chi (z^2)\gamma_\chi (w^2) =:\beta_\chi (z^2)\gamma_\chi (w^2): +\frac{1}{z^2-w^2},
\]
and so
\begin{align*}
H^{\beta} (z^2) H^{\gamma} (w^2)& = i_{z, w}\frac{1}{z^2 -w^2} V^+(z)^{-1}V^+(w)\beta_\chi (z^2)\gamma_\chi (w^2)z^{-2h_0}w^{2h_0}V^-(z)^{-1}V^-(w)\\
&= i_{z, w}\frac{1}{z^2 -w^2} V^+(z)^{-1}V^+(w)\left(:\beta_\chi (z^2)\gamma_\chi (w^2): +\frac{1}{z^2-w^2}\right) z^{-2h_0}w^{2h_0}V^-(z)^{-1}V^-(w)\\
&= i_{z, w}\frac{1}{z^2 -w^2} V^+(z)^{-1}V^+(w)\left(:\beta_\chi (z^2)\gamma_\chi (w^2):\right) z^{-2h_0}w^{2h_0}V^-(z)^{-1}V^-(w)\\
 &\hspace{1cm} + i_{z, w}\frac{1}{(z^2 -w^2)^2} V^+(z)^{-1}V^+(w)z^{-2h_0}w^{2h_0}V^-(z)^{-1}V^-(w).
\end{align*}
We can similarly derive
\begin{align*}
H^{\gamma} (w^2) H^{\beta} (z^2) & =  V^+(w)\gamma_\chi (w^2)w^{2h_0}V^-(w) V^+(z)^{-1}\beta_\chi (z^2)z^{-2h_0}V^-(z)^{-1} \\
&= \frac{w^2 -z^2}{w^2} V^+(w)\gamma_\chi (w^2)w^{2h_0} V^+(z)^{-1} V^-(w) \beta_\chi (z^2)z^{-2h_0}V^-(z)^{-1}\\
&= \frac{w^2 -z^2}{w^2}\cdot i_{w, z}\frac{w^2}{w^2 -z^2} V^+(w)\gamma_\chi (w^2)w^{2h_0} V^+(z)^{-1}  \beta_\chi (z^2)z^{-2h_0}V^-(w)V^-(z)^{-1}\\
&= \frac{w^2 -z^2}{w^2}\cdot i_{w, z}\frac{w^2}{w^2 -z^2}\cdot i_{w, z}\frac{w^2}{w^2 -z^2} V^+(w) V^+(z)^{-1} \gamma_\chi (w^2)w^{2h_0} \beta_\chi (z^2)z^{-2h_0}V^-(w)V^-(z)^{-1}\\
&= i_{w, z}\frac{1}{w^2 -z^2} V^+(w) V^+(z)^{-1} \gamma_\chi (w^2) \beta_\chi (z^2)z^{-2h_0}w^{2h_0}V^-(w)V^-(z)^{-1}\\
&= i_{w, z}\frac{1}{w^2 -z^2} V^+(w) V^+(z)^{-1} \left(:\gamma_\chi (w^2) \beta_\chi (z^2): -i_{w, z}\frac{1}{w^2-z^2}\right)z^{-2h_0}w^{2h_0}V^-(w)V^-(z)^{-1}\\
&= i_{w, z}\frac{1}{w^2 -z^2} V^+(w) V^+(z)^{-1} \left(:\beta_\chi (z^2) \gamma_\chi (w^2) : \right)z^{-2h_0}w^{2h_0}V^-(w)V^-(z)^{-1}\\
&\hspace{1cm}  - i_{w, z}\frac{1}{(w^2 -z^2)^2} V^+(w) V^+(z)^{-1}z^{-2h_0}w^{2h_0}V^-(w)V^-(z)^{-1}
\end{align*}
Thus we have
\begin{align*}
H^{\beta} (z^2) H^{\gamma} (w^2) &+ H^{\gamma} (w^2) H^{\beta} (z^2) =\delta (z^2, w^2) V^+(w) V^+(z)^{-1} \left(:\beta_\chi (z^2) \gamma_\chi (w^2) : \right)z^{-2h_0}w^{2h_0}V^-(w)V^-(z)^{-1}\\
& \hspace{1cm} +\partial_{w^2}\delta (z^2, w^2)  V^+(w) V^+(z)^{-1}z^{-2h_0}w^{2h_0}V^-(w)V^-(z)^{-1}
\end{align*}
Now we use the standard properties of the the delta function (see e.g. \cite{Kac}), namely
\[
\delta (z^2, w^2) f(z^2) =\delta (z^2, w^2) f(w^2), \quad \text{and}\quad  \partial_{w^2}\delta (z^2, w^2) f(z^2) = \partial_{w^2}\delta (z^2, w^2)f(w^2) + \delta (z^2, w^2)\partial_{w^2} f(w^2).
\]
Consequently,
\begin{align*}
H^{\beta} (z^2) H^{\gamma} (w^2) + H^{\gamma} (w^2) H^{\beta} (z^2) & =\delta (z^2, w^2) :\beta_\chi (w^2) \gamma_\chi (w^2): +\\
& \hspace{0.7cm} +\partial_{w^2}\delta (z^2, w^2)  + \delta (z^2, w^2)\left(-\sum_{n>0} h_{-n} w^{2n-2} -\sum_{n>0} h_n w^{-2n-2}-h_0w^{-2}\right)\\
& = \delta (z^2, w^2) h_\chi(w) +\partial_{w^2}\delta (z^2, w^2) - \delta (z^2, w^2) h_\chi(w)\\
& = \partial_{w^2}\delta (z^2, w^2).
\end{align*}
Now we prove the first of the trivial relations:
\begin{align*}
H^{\beta} (z^2) H^{\beta} (w^2)& =  V^+(z)^{-1}\beta_\chi (z^2)z^{-2h_0}V^-(z)^{-1}  V^+(w)^{-1}\beta_\chi (w^2)w^{-2h_0}V^-(w)^{-1}\\
&= i_{z, w}\frac{z^2}{z^2 -w^2}  V^+(z)^{-1}\beta_\chi (z^2)  V^+(w)^{-1} z^{-2h_0}  V^-(z)^{-1}\beta_\chi (w^2)w^{-2h_0}V^-(w)^{-1}\\
&= i_{z, w}\frac{z^2}{z^2 -w^2}\cdot \frac{z^2 -w^2}{z^2}  V^+(z)^{-1}  V^+(w)^{-1} \beta_\chi (z^2)  z^{-2h_0}  V^-(z)^{-1}\beta_\chi (w^2)w^{-2h_0}V^-(w)^{-1}\\
&= i_{z, w}\frac{z^2}{z^2 -w^2}\cdot \frac{z^2 -w^2}{z^2}\cdot \frac{z^2 -w^2}{z^2}  V^+(z)^{-1}  V^+(w)^{-1} \beta_\chi (z^2)  z^{-2h_0} \beta_\chi (w^2)w^{-2h_0} V^-(z)^{-1} V^-(w)^{-1}\\
&= (z^2 -w^2)  V^+(z)^{-1}  V^+(w)^{-1} \beta_\chi (z^2)  \beta_\chi (w^2) z^{-2h_0} w^{-2h_0} V^-(z)^{-1} V^-(w)^{-1}.
\end{align*}
Therefore
\[
H^{\beta} (w^2) H^{\beta} (z^2)= (w^2 -z^2)  V^+(z)^{-1}  V^+(w)^{-1} \beta_\chi (w^2)  \beta_\chi (z^2) z^{-2h_0} w^{-2h_0} V^-(z)^{-1} V^-(w)^{-1},
\]
and so
\[
H^{\beta} (z^2) H^{\beta} (w^2)+ H^{\beta}(w) H^{\beta} (z^2) =0.
\]
The relation
\[
H^{\gamma} (z^2) H^{\gamma} (w^2)+ H^{\gamma}(w) H^{\gamma} (z^2) =0.
\]
is proved similarly.
\end{proof}
We index the fields $H^{\beta} (z^2)$ and $H^{\gamma} (z^2)$ in the standard vertex algebra notation:
\begin{equation}
H^{\beta} (z^2)=\sum_{n\in \mathbb{Z}} H^{\beta}_{(n)} z^{-2n-2}; \quad H^{\gamma} (z^2) = \sum_{n\in \mathbb{Z}} H^{\gamma}_{(n)} z^{-2n-2}.
\end{equation}
The proposition above ensures that  $H^{\beta} (z)$ and $H^{\gamma} (z)$  satisfy the OPE relations of the symplectic fermion vertex algebra introduced by Kausch, see e.g. \cite{Curiosities} and \cite{SymplFermionsKausch}. Observe that since the fields $H^{\beta} (z^2)$ and $H^{\gamma} (z^2)$ depend only on $z^2$ we can re-scale back to $z$ as is necessary for a super vertex algebra.  Now we need to check that the space $\mathit{F^{hwv}_\chi}$ satisfies the other conditions for the existence of a vertex algebra structure, as in e.g. the Existence Theorem 4.5, \cite{Kac}.
It is immediate to check that the creation condition is satisfied:
\begin{equation}
H^{\beta} (z^2) |0\rangle  = V^+(z)^{-1}\beta_\chi(z^2)|0\rangle = \chi _{-3/2}|0\rangle  +O(z^2),
\end{equation}
and
\begin{equation}
 H^{\gamma} (z^2) |0\rangle  = V^+(z)\gamma_\chi(z^2)|0\rangle = \chi _{-1/2}|0\rangle  +O(z^2).
\end{equation}
In order to show that the the operators $H^{\beta}_{(n)}$ and $H^{\gamma}_{(n)}$ generate the vector space $\mathit{F^{hwv}_\chi}$ by a successive action on the vacuum $|0\rangle $, we observe that the
vector
\[
H_{(n_1)}H_{(n_2)}\dots H_{(n_k)}|0\rangle,
\]
where $H_{(n_s)}$ is either $H^{\beta}_{(n_s)}$ or $H^{\gamma}_{(n_s)}$,
will appear as a coefficient in the multivariable expression
\[
H(z_1^2)H(z_2^2)\dots H(z_k^2)|0\rangle,
\]
where again $H(z_s^2)$ is either $H^{\beta}(z_s^2)$ or $H^{\gamma}(z_s^2)$. We first observe that as a consequence of
Lemma \ref{lem:H-betawithHeis} these coefficients are themselves highest  weight vectors for the Heisenberg action.

By extending the calculation in the proof of the previous proposition, we can see that
\[
H(z_1^2)H(z_2^2)\dots H(z_k^2)=\prod_{s>l} i_{z_s, z_l}(z_s^2 -z_l^2)^\pm \prod_{s=1}^k \left(V^+(z_s)^\pm\right) \prod_{s=1}^k(\beta-\text{or}-\gamma)_\chi (z_s^2)\prod_{s=1}^k \left(z_s^{\pm2h_0} V^-(z_s)^{\pm}\right),
\]
the $\pm$ depends on whether the  $H(z_s^2)$ is  $H^{\beta}(z_s^2)$ or $H^{\gamma}(z_s^2)$.
Therefore we have
\[
H(z_1^2)H(z_2^2)\dots H(z_k^2)|0\rangle=\prod_{s>l} i_{z_s, z_l}(z_s^2 -z_l^2)^\pm \prod_{s=1}^k \left(V^+(z_s)^\pm\right) \prod_{s=1}^k(\beta-\text{or}-\gamma)_\chi (z_s^2)|0\rangle.
\]
Now the nonzero coefficients in the above multivariate expression will be precisely those for which the coefficients in
$\prod_{s=1}^k(\beta-\text{or}-\gamma)_\chi (z_s^2)|0\rangle$ cannot be canceled by an action of the operators from $\prod_{s=1}^k \left(V^+(z_s)^\pm\right)$. The coefficients in $\prod_{s=1}^k(\beta-\text{or}-\gamma)_\chi (z_s^2)|0\rangle$  are the elements $\left(\chi _{-j_k}\right)^{m_k}\dots \left(\chi _{-j_2}\right)^{m_2}\left(\chi _{-j_1}\right)^{m_1}|0\rangle$ and they span $\mathit{F_\chi}$. Thus the nonzero coefficients will correspond precisely to monomials  $\left(\chi _{-j_k}\right)^{m_k}\dots \left(\chi _{-j_2}\right)^{m_2}\left(\chi _{-j_1}\right)^{m_1}|0\rangle$  that cannot be obtained by acting with the Heisenberg algebra on combinations of similar monomials but of lower degree. Due to the fact that the representation of the Heisenberg algebra on $\mathit{F_\chi}$ is completely reducible, those correspond  precisely to the highest weight vectors for the Heisenberg action.   Thus we see that successive action by the operators $H^{\beta}_{(n)}$ and $H^{\gamma}_{(n)}$ will generate the the space $\mathit{F^{hwv}_\chi}$ of the highest weight vectors for the Heisenberg action. In fact we can see directly that this is a strong generation, i.e, the only indexes appearing in the generating elements $H_{(n_1)}H_{(n_2)}\dots H_{(n_k)}|0\rangle$  are negative, $n_s <0, s=1, 2, \dots k$.

\begin{example}  For the two special families of highest weight vectors
$ \chi_{-\frac{1}{2}}^n|0\rangle$, \ $\chi_{-\frac{3}{2}}^n|0\rangle$  one can easily check that
\begin{align}
\chi_{-\frac{1}{2}}^n|0\rangle &= H^{\gamma}_{(-n)}\dots H^{\gamma}_{(-2)} H^{\gamma}_{(-1)}|0\rangle\\
\chi_{-\frac{3}{2}}^n|0\rangle &= H^{\beta}_{(-n)}\dots H^{\beta}_{(-2)} H^{\beta}_{(-1)}|0\rangle.
\end{align}
\end{example}
Finally, to apply the Existence Theorem 4.5 of \cite{Kac} we need a Virasoro element, which will define the translation operator. As is well known, from the start the symplectic vertex algebra was of interest due to the  properties of its Virasoro field and its (logarithmic) modules. Namely, it is immediate to calculate that the field (observe that on the space  $\mathit{F^{hwv}_\chi}$ this normal ordered product is well defined):
\begin{equation}
L^{hwv} (z^2):  =: H^{\gamma}(z^2) H^{\beta}(z^2):  =\sum_{n\in \mathbb{Z}} L^{hwv}_n z^{-2n-4}
\end{equation}
is a Virasoro field with central charge $c=-2$, namely
\[
L^{hwv} (z^2)L^{hwv} (w^2)\sim \frac{2L^{hwv} (w^2)}{(z^2-w^2)^2} + \frac{\partial_{w^2}L^{hwv} (w^2)}{z^2-w^2} -\frac{1}{(z^2-w^2)^4}.
\]
This can  easily be proved by Wick's Theorem using the OPEs derived in Proposition \ref{prop:symOPEs}, so we omit it. Thus we can take $L^{hwv}_{-1}$ as a translation operator $T$ on $\mathit{F^{hwv}_\chi}$.
We can then immediately calculate that
\[
T|0\rangle =0, \quad [T, H^{\beta}(z^2)] =\partial_{z^2} H^{\beta}(z^2), \quad \text{and}\quad [T, H^{\gamma}(z^2)] =\partial_{z^2} H^{\gamma}(z^2),
\]
which completes the requirements of  the Existence Theorem 4.5 of \cite{Kac}.
Thus, and after observing that we can re-scale from   $z^2$ to $z$ (as all relevant fields, namely $H^{\beta}(z^2)$, $H^{\gamma}(z^2)$ and $L^{hwv} (z^2)$, depend only on $z^2$),  we arrive at the following
\begin{thm}\label{thm:symplVA}
The vector space  $\mathit{F^{hwv}_\chi}$ spanned by the highest weight vectors has a structure of a super vertex algebra,  strongly generated by the fields $H^{\beta}(z)$ and $H^{\gamma}(z)$, with vacuum vector $|0\rangle$, translation operator $T=L^{hwv}_{-1}$, and vertex operator map induced by \begin{equation}
Y( \chi _{-1/2}|0\rangle, z)=H^{\gamma} (z), \quad Y( \chi _{-3/2}|0\rangle, z)=H^{\beta} (z).
\end{equation}
This vertex algebra structure is a realization of  the symplectic fermion vertex algebra, indicated by the OPEs:
\begin{align}
&H^{\beta} (z) H^{\gamma} (w) \sim \frac{1}{(z -w)^2},  \quad  H^{\gamma} (z)H^{\beta} (w)\sim -\frac{1}{(z -w)^2}; \\
&H^{\beta} (z) H^{\beta} (w)\sim  0; \quad H^{\gamma} (z) H^{\gamma} (w) \sim 0.
\end{align}
\end{thm}
\section{Complete  bosonization}
As we saw in Section 3, the fields $\beta_\chi (z^2)$ and $\gamma_\chi (z^2)$ needed to express the generating field
\[
\chi(z) =\gamma_\chi(z^2) +z \beta_\chi (z^2)
\]
can be written as
\begin{equation} \label{eqn:beta-gamma-bos}
\beta_\chi(z^2)=V^+(z)H^{\beta} (z^2)V^-(z)z^{2h_0}, \quad \quad  \gamma_\chi (z^2) =V^+(z)^{-1}H^{\gamma} (z^2)V^-(z)^{-1}z^{-2h_0}.
\end{equation}
Due to Proposition \ref{prop:HeisDecomp} we can write
\begin{equation}
\mathit{F_\chi}\cong \mathit{F^{hwv}_\chi}\ten \mathbb{C}[x_1,
x_2, \dots , x_n, \dots ].
\end{equation}
The fields    $V^+(z)$ and $V^-(z)$ (consequently $V^+(z)^{-1}$ and $V^-(z)^{-1}$) are bosonic,
 via the action:
\begin{equation} \label{eqn:Heisendiff-x}
h_n\mapsto i\partial _{x_{n}}, \quad h_{-n} \mapsto
inx_n\cdot, \quad \text{for any} \ \ n>0.
\end{equation}
\begin{remark}
As we mentioned before,  we can use an arbitrary re-scaling $h_n\to s_n h_n, \ s_n\neq 0, \ s_n\in \mathbb{C}$, for $n >0$, so that we could have used instead the identification
\[
h_n\mapsto -\partial _{x_{n}}, \quad h_{-n} \mapsto
nx_n\cdot, \quad \text{for any} \ \ n\in \mathbb{N}.
\]
The identification we use here underlines the potential complexification, as seen in \eqref{eqn:finalcompl1} and \eqref{eqn:finalcompl2} below.
\end{remark}
But the fields $H^{\beta} (z^2)$ and $H^{\gamma} (z^2)$ required to complete the description of the generating field $\chi(z)$ are fermionic.
We can, as was done in  \cite{OrlovLeur} for the twisted bosonization, introduce super-variables and derivatives with respect to those super variables to describe  the fields  $H^{\beta} (z)$ and $H^{\gamma} (w)$ and their action on the  space of the highest weigh vectors $\mathit{F^{hwv}_\chi}$.  But in this case, for this second bosonization we can do better, as it is known that the symplectic fermions can be embedded into a lattice vertex algebra. Namely, as in the  Friedan-Martinec-Shenker (FMS) bosonization, \cite{FMS}, and following \cite{Curiosities} and \cite{SymplFermionsKausch}, we can view the fields $H^{\beta} (z)$ and $H^{\gamma} (w)$ as
\begin{equation}
H^{\beta} (z) \mapsto \psi^- (z), \quad H^{\gamma} (z) \mapsto \partial_z\psi^+ (z);
\end{equation}
where $\psi^+ (z)$ and $\psi^- (z)$  are the  charged free  fermion fields used in the bosonization of the KP hierarchy, via the boson-fermion correspondence (see e.g. \cite{Kac}, \cite{Miwa-book}). Specifically, $\psi^+ (z)$ and $\psi^- (z)$ have OPEs
\[
\psi^+ (z)\psi^- (w)\sim \frac{1}{z-w}, \quad \psi^- (z)\psi^+ (w)\sim \frac{1}{z-w}, \quad \psi^+ (z)\psi^+ (w)\sim 0, \quad \psi^- (z)\psi^- (w)\sim 0,
\]
and are the generating fields of the  charged free  fermion vertex super algebra (see e.g \cite{Kac}).
We can use the bosonization of  the  charged free fermion vertex super algebra via the lattice fields
\begin{equation}
\psi^+ (z)\to e^{\alpha}_y (z), \quad \psi^+ (z)\to e^{-\alpha}_y (z),
\end{equation}
where the lattice fields $e^{\alpha}_y (z)$, $e^{-\alpha}_y (z)$  act on
the bosonic  vector space $\mathbb{C}[e^{\alpha}, e^{-\alpha}]\ten \mathbb{C}[y_1, y_2, \dots , y_n\dots]$  by
\begin{align*}
e^{\alpha}_y(z)& =\exp (\sum _{n\ge 1} y_n z^n)\exp (-\sum _{n\ge 1}\frac{\partial}{n\partial y_n} z^{-n})e^{\alpha}z^{\partial_{\alpha}},\\
e^{-\alpha}_y(z)& =\exp (-\sum _{n\ge 1}y_n z^n)\exp (\sum _{n\ge 1}\frac{\partial}{n\partial y_n} z^{-n})e^{-\alpha}z^{-\partial_{\alpha}},
\end{align*}
as is standard in the theory of the KP hierarchy. We use the index $y$
 to indicate these are  the exponentiated boson fields acting on the variables $y_1, y_2, \dots , y_n\dots$. We introduce similarly the Heisenberg field $h^y (z)$,
\begin{equation}
h^y (z) = \sum _{n\ge 1}\frac{\partial}{\partial y_n} z^{-n-1} +h^y_0 z^{-1} + \sum _{n\ge 1}n y_n z^{n-1},
\end{equation}
where $h^y_0$ acts on $\mathbb{C}[e^{\alpha}, e^{-\alpha}]\ten \mathbb{C}[y_1, y_2, \dots , y_n\dots]$  by
$h^y_0 e^{m\alpha} P(y_1, y_2, \dots , y_n\dots) = m e^{m\alpha} P(y_1, y_2, \dots , y_n\dots)$.
Thus, combining the two maps,  we  map $\mathit{F^{hwv}_\chi}$ into a subspace of $\mathbb{C}[e^{\alpha}, e^{-\alpha}]\ten \mathbb{C}[y_1, y_2, \dots , y_n\dots]$, and
\begin{align}\label{eqn:ident1}
H^{\beta} (z) &\to e^{-\alpha}_y(z) =\exp (-\sum _{n\ge 1}y_n z^n)\exp (\sum _{n\ge 1}\frac{\partial}{n\partial y_n} z^{-n})e^{-\alpha}z^{-h^y_0}\\ \label{eqn:ident2}
H^{\gamma} (z) &\to \partial_z e^{\alpha}_y(z) =:h^y (z)\exp (\sum _{n\ge 1} y_n z^n)\exp (-\sum _{n\ge 1}\frac{\partial}{n\partial y_n} z^{-n})e^{\alpha}z^{h^y_0}:.
\end{align}
Now we can combine the actions of the two Heisenberg fields, the $h^y (z)$ and the original $h_\chi(z)$.  Through the above map  the Fock space $\mathit{F_{\chi}}$ will be mapped to a subspace of $\mathbb{C}[e^{\alpha}, e^{-\alpha}]\ten \mathbb{C}[x_1, x_2, \dots , x_n, \dots ;y_1, y_2, \dots , y_n\dots]$. The modes $h_n$ (for clarity we shall write $h^x_n$) of the field $h_\chi(z)$ will act as in \eqref{eqn:Heisendiff-x}, with
\begin{equation}
h^x_0 e^{m\alpha} P(x_1, x_2, \dots , x_n, \dots ;y_1, y_2, \dots , y_n\dots) =m e^{m\alpha} P(x_1, x_2, \dots , x_n, \dots ;y_1, y_2, \dots , y_n\dots).
\end{equation}
The action of $h^x_0$  stems from the identifications \eqref{eqn:ident1} and \eqref{eqn:ident2}  which  determine the charges of the elements of $\mathbb{C}[e^{\alpha}, e^{-\alpha}]\ten \mathbb{C}[y_1, y_2, \dots , y_n\dots]$. Thus implies that the actions  $z^{-h^x_0}$ in \eqref{eqn:beta-gamma-bos} and $z^{h^y_0}$ in  \eqref{eqn:ident2}  will cancel each other.
And so finally we arrive at the complete second bosonization of the CKP hierarchy:
\begin{thm}
The generating field $\chi (z)$ of the CKP hierarchy can be written as
\[
\chi(z) =\gamma_\chi(z^2) +z \beta_\chi (z^2),
\]
where the fields $\beta_\chi (z)$ and $\gamma_\chi(z)$ can be bosonized as follows:
\begin{align} \label{eqn:finalcompl1}
\beta_\chi (z) &\to  \ \exp \Big(i\sum_{n>0}(x_n +iy_n) z^{n}\Big)\exp \Big(-i\sum_{n>0}\frac{1}{n}\left(\frac{\partial}{\partial x_n}+i\frac{\partial}{\partial y_n}\right) z^{-n}\Big)e^{-\alpha},\\ \label{eqn:finalcompl2}
\gamma_\chi(z) &\to \  :\exp \Big(-i\sum_{n>0}(x_n +iy_n)z^{n}\Big) h^y (z) \exp \Big(i\sum_{n>0}\frac{1}{n}\left(\frac{\partial}{\partial x_n}+i\frac{\partial}{\partial y_n}\right) z^{-n}\Big) e^{\alpha}:.
\end{align}
The Fock space $\mathit{F_{\chi}}$ is mapped to a subspace of the bosonic space $\mathbb{C}[e^{\alpha}, e^{-\alpha}]\ten \mathbb{C}[x_1, x_2, \dots , x_n, \dots ;y_1, y_2, \dots , y_n\dots]$, with \ $|0\rangle \mapsto 1$.
The Hirota equation \eqref{eqn:Hirotaeqn} is equivalent to
\begin{equation}
Res_z \Big(\beta_\chi (z)\otimes \gamma_\chi (z)- \gamma_\chi (z)\otimes \beta_\chi (z)\Big) =0.
\end{equation}
\end{thm}

\section{Outlook}
In this paper we completed the second bosonization of the Hirota equation for the CKP hierarchy. Here we did not discuss solutions, symmetries, complexification, nor  further applications, such as certain character formulas, vacuum expectation values equalities, etc. Also, the consequences of the existence of the two bosonizations (the one described here as well as the bosonization studied in \cite{OrlovLeur}), need to be addressed, as well as the comparison between the Hirota equation and the reduction approach to the CKP hierarchy.  Each of these topics is  worth a separate discussion, which we will commence in a consequent paper.

\def\cprime{$'$}

 \end{document}